\documentclass[reqno]{amsproc}
\usepackage{amssymb}
\usepackage{amsmath}
\usepackage{amsfonts}
\usepackage{color}
\usepackage{geometry}
\usepackage{bbm}
\usepackage{stmaryrd}
\usepackage{cite}
\usepackage{tikz}

\definecolor{myurlcolor}{rgb}{0,0,0.4}
\definecolor{mycitecolor}{rgb}{0,0.5,0}
\definecolor{myrefcolor}{rgb}{0.5,0,0}
\usepackage[pagebackref,draft=false]{hyperref}
\hypersetup{colorlinks,
linkcolor=myrefcolor,
citecolor=mycitecolor,
urlcolor=myurlcolor}

\usepackage[capitalize]{cleveref}

\newcommand{\beq}{\begin{equation}}
\newcommand{\eeq}{\end{equation}}
\newcommand{\Z}{\mathbb{Z}}
\newcommand{\N}{\mathbb{N}}

\newcommand{\C}{\mathbb{C}}
\renewcommand{\H}{\mathcal{H}}
\newcommand{\B}{\mathcal{B}}
\newcommand{\U}{\mathcal{U}}
\newcommand{\Cl}{\mathcal{C}}
\newcommand{\defin}{:=}
\newcommand{\myand}{\mathrel{\mathsf{and}}}
\newcommand{\myimplies}{\:\mathrel{\mathsf{implies}}\:}
\DeclareMathOperator*{\mybigand}{\mathsf{And}\,}
\newcommand{\PU}[1]{OC(#1)}

\theoremstyle{plain}
\newtheorem{dummy}{Dummy}
\newtheorem{thm}[dummy]{Theorem}
\newtheorem*{thm*}{Theorem}
\newtheorem{lem}[dummy]{Lemma}

\newtheorem{cor}[dummy]{Corollary}
\newtheorem*{cor*}{Corollary}

\newtheorem{defn}[dummy]{Definition}
\newtheorem{fact}[dummy]{Fact}

\theoremstyle{remark}
\newtheorem{ex}[dummy]{Example}
\newtheorem{rem}[dummy]{Remark}

\allowdisplaybreaks

\begin{document}
\sloppy

\setlength{\jot}{6pt}



\title{Quantum logic is undecidable}

\author{Tobias Fritz}
\email{tfritz@pitp.ca}
\address{Perimeter Institute for Theoretical Physics, Waterloo, Canada}



\begin{abstract}
We investigate the first-order theory of closed subspaces of complex Hilbert spaces in the signature $(\lor,\perp,0,1)$, where `$\perp$' is the orthogonality relation. Our main result is that already its quasi-identities are undecidable: there is no algorithm to decide whether an implication between equations and orthogonality relations implies another equation. This is a corollary of a recent result of Slofstra in combinatorial group theory. It follows upon reinterpreting that result in terms of the hypergraph approach to quantum contextuality, for which it constitutes a proof of the \emph{inverse sandwich conjecture}. It can also be interpreted as stating that a certain quantum satisfiability problem is undecidable.
\end{abstract}

\maketitle

\section{Introduction}

Quantum logic starts with the idea that quantum theory can be understood as a theory of physics in which standard Boolean logic gets replaced by a different form of logic, where various rules, such as the distributivity of logical and over logical or, are relaxed~\cite{logicqm,handbook}. This builds on the observation that $\{0,1\}$-valued observables behave like logical propositions: such an observable is a projection operator on Hilbert space, and it can be identified with the closed subspace that it projects onto. In this way, the conjunction (logical \emph{and}) translates into the intersection of subspaces, while disjunction (logical \emph{or}) is interpreted as forming the closed subspace spanned by two subspaces. Hence the closed subspaces of a complex Hilbert space $\H$ form the \emph{complex Hilbert lattice} $\Cl(\H)$, which is interpreted as the lattice of `quantum propositions' and forms a particular kind of orthomodular lattice~\cite{orthomodular,meashilb,ql}.

However, the theory of orthomodular lattices is quite rich and contains many objects other than complex Hilbert lattices. So in order to understand the laws of quantum logic, one has to find additional properties which characterize the latter kind of objects. Much effort has been devoted to this question, resulting in partial characterizations such as Piron's theorem~\cite{piron,pironfix}, Wilbur's theorem~\cite{wilbur} and Sol\`er's theorem~\cite{soler}\footnote{See also the survey~\cite{hilbgeom} for a more recent exposition from a geometrical perspective.}. However, the axioms for complex Hilbert lattices that these results suggest are quite sophisticated: atomicity, completeness or the existence of an infinite orthonormal sequence. These are conditions that cannot be expressed algebraically, i.e.~as \emph{first-order} properties using just a finite number of variables, algebraic operations, and quantifiers. Fortunately, there has also been a substantial amount of work on first-order properties enjoyed by complex Hilbert lattices, and in particular on equational laws that hold in all complex Hilbert lattices $\Cl(\H)$, such as the algorithmic approach advocated by Megill and Pavi{\v{c}}i\'c~\cite{hleqs,qlexpl}. Such algorithmic approaches are what our present contribution is about: we prove that there cannot exist any algorithm to decide whether an \emph{implication} between equations in complex Hilbert lattices holds for all variable assignments.

To make this statement precise, we keep the lattice operations notationally separate from the external logical connectives and denote the latter in plain English. We consider the theory of Hilbert lattices in the signature $(\lor,\perp,0,1)$, where is the binary lattice join operation, $\perp$ is the binary orthogonality relation, and $0$ and $1$ are constants denoting the bottom and top lattice elements respectively, corresponding to the zero subspace and the full subspace.

The following result is an immediate consequence of our \Cref{sandwich} presented in \Cref{mainthm}.

\begin{thm}
\label{qlundec}
For a complex Hilbert space $\H$, let $\Cl(\H)$ denote the lattice of projections on $\H$ or equivalently closed subspaces of $\H$. Let $0,1\in\Cl(\H)$ stand for the zero subspace and full subspace, respectively. Let $P_1,P_2,\ldots$ be free variables taking values in $\Cl(\H)$.

Then there is no algorithm to decide whether a sentence of the form
\beq
\label{Eprop}
	\forall P_1,\ldots, P_n \,\left[ \left( E_1 \myand E_2 \myand \ldots \myand E_k \right) \myimplies (0 = 1)\right]
\eeq
holds in $\Cl(\H)$ for every $\H$, where each $E_i$ is a formula having one of the following two forms:
\begin{itemize}
	\item an equation of the form $P_{i_1}\lor\ldots\lor P_{i_k} = 1$;
	\item an orthogonality relation $P_{i_1}\perp P_{i_2}$ between two free variables.
\end{itemize}
\end{thm}

Here, the consequent $0=1$ states that the zero subspace is equal to the whole Hilbert space, or equivalently that $\H$ is the trivial zero-dimensional Hilbert space. In other words, the sentence~\eqref{Eprop} states that the antecedents $E_1,\ldots,E_k$ are jointly contradictory in any nonzero Hilbert space. Note that all sentences of the form~\eqref{Eprop} are \emph{quasi-identities}~\cite[p.~149]{malcev}.

\begin{rem}
	It may also be interesting to work with orthocomplementation $P \mapsto P^\perp$ as a unary operation instead of the binary relation $\perp$, as usually done in the theory of orthomodular lattices. In this signature, our result is the same, with $P \perp Q$ replaced by $P \land Q^\perp = P$, which is semantically equivalent in all $\Cl(\H)$.
	
	In this new signature, we could also take orthocomplements everywhere, replacing each $P_i$ by $P_i^\perp$. In this way, it follows that Theorem~\ref{qlundec} remains true if one replaces $\lor$ by $\land$ and $1$ by $0$, so that each $E_i$ is either of the form $P_{i_1}\land\ldots\land P_{i_k} = 0$ or $P_{i_1} \land P_{i_2}^\perp = P_{i_1}$.
\end{rem}

\begin{rem}
Instead of asking whether~\eqref{Eprop} holds in all $\H$, one can alternatively negate the question and ask whether there exists a Hilbert space $\H$ with $\dim(\H) > 0$ together with an assignment of projections in $\H$ to the free variables such that the formula
\[
	E_1 \myand E_2 \myand \ldots \myand E_k
\]
holds. This formulation makes it clear that we are dealing with a quantum version of the Boolean satisfiability problem---distinct from the QSAT problem introduced by Bravyi~\cite{QSAT}---which is undecidable as per Theorem~\ref{qlundec}. As we will see in the proof of Lemma~\ref{allH}, if an instance of our quantum satisfiability problem is solvable, then it is also solvable with $\H$ infinite-dimensional separable. Thus it is sufficient to consider e.g.~$\H = \ell^2(\N)$ only.

The reason that we prefer the statement of Theorem~\ref{qlundec} over the satisfiability formulation is that we are interested in the laws of quantum logic, i.e.~in those statements that hold in \emph{all} Hilbert spaces $\H$.
\end{rem}

\begin{ex}
\label{triangle}
The implication
\begin{align*}
(P\lor Q = 1) \myand {} & (Q\lor R = 1) \myand (R\lor P = 1) \\
& \myand (P\perp Q) \myand (Q\perp R) \myand (R\perp P) \myimplies (0=1)
\end{align*}
is valid: in any nonzero Hilbert space, it is impossible to have three projections that are pairwise orthogonal and such that any two of them sum to the identity~\cite{spparable}.
\end{ex}

The key ingredient that leads to Theorem~\ref{qlundec} is an undecidability result of Slofstra~\cite{tp}, who builds on earlier work of Cleve, Liu and Slofstra~\cite{solgroup} and Cleve and Mittal~\cite{bcsg}. Our contribution merely consists of having seen the connection to quantum logic via the hypergraph approach to contextuality~\cite{AFLS}. The mathematical depth necessary for deriving such an undecidability result is to be found in Slofstra's arguments.

The statement that we actually prove first is Corollary~\ref{isc}, which is the \emph{inverse sandwich conjecture} from~\cite{AFLS}. The undecidability of quantum logic in the form of Theorem~\ref{sandwich} is then merely a reformulation---on an even smaller set of sentences than our formulation above. As we will see, the statement remains true if one replaces `holds in every $\Cl(\H)$' by `holds in $\Cl(\H)$ for some infinite-dimensional separable Hilbert space $\H$'. 

Corollary~\ref{isc} also implies that infinitely many of the free hypergraph C*-algebras $C^*(H)$ of~\cite{AFLS} fail to be residually finite-dimensional, as per Corollary~\ref{notrfd}.

\subsection*{Related work} 

Lipshitz~\cite{modunsolv} has shown, among other things, that the purely implicational fragment of the theory of all $\Cl(\C^n)$ is undecidable, already in the signature $(\lor,\land,0,1)$. While this result is similar to ours, it uses techniques specific to a finite-dimensional setting, namely coordinatization. 
A result of Sherif~\cite{orthodec} is that any first-order theory between orthomodular lattices and finite orthomodular lattices is undecidable.
Herrmann~\cite{vna} has proven that the equational theory of the orthomodular lattice of projections of a finite von Neumann algebra factor is decidable; this includes both the $\Cl(\C^n)$ and the projection ortholattices of factors of type $\mathrm{II}_1$. Other work of Herrmann and Ziegler is also concerned with related decidability and complexity problems~\cite{quantumsat}.

After the preprint version of this paper appeared, Atserias, Kolaitis and Severini~\cite{gensat} investigated general classes of satisfiability problems with variables in the projection lattice of Hilbert space. Based on Slofstra's results together with the developments of this paper, they proved a sharp separation result analogous to Schaefer's classical dichotomy theorem on Boolean satisfiability.

Based on the present results, we have investigated further undecidable properties of free hypergraph C*-algebras in the paper~\cite{freehyper}, which is a follow-up to the present one.

\section{Solution groups and their group C*-algebras}

Before getting to the proof of our Theorem~\ref{qlundec}, we review the essential ingredient: Slofstra's recent work in combinatorial group theory~\cite{tp}. Subsequently, we will modify his intended interpretation using nonlocal games to one in terms of contextuality. From there, it is only a small step to quantum logic.

Following~\cite{bcsg}, Slofstra considers \emph{linear systems} over $\Z_2$, which are linear equations $Mx = b$ with $M\in\Z_2^{m\times n}$ and $b\in\Z_2^m$. While conventional solutions have $x\in\Z_2^n$, a \emph{quantum solution}~\cite{solgroup} consists of self-adjoint operators $A_1,\ldots,A_n\in\B(\H)$ such that:
\begin{itemize}
\item $A_i^2 = \mathbbm{1}$ for all $i$;
\item If $x_i$ and $x_j$ appear in the same equation, then $A_i$ commutes with $A_j$;
\item For each equation of the form $x_{k_1} + \ldots + x_{k_r} = b_r$, the operators satisfy
\[
	A_{k_1} \cdots A_{k_r} = (-1)^{b_r}\mathbbm{1},
\]
where the order of the factors is irrelevant due to the previous commutativity requirement.
\end{itemize}

The fact that quantum solutions solve the given linear system multiplicatively instead of additively is purely conventional, and allows for simpler notation. The most famous example of a quantum solution of a linear system that is unsolvable over $\Z_2$ is the Mermin-Peres magic square~\cite{peres}. The quantum solution for $\H = \C$ are precisely the conventional solutions over $\Z$, written multiplicatively as $A_i = (-1)^{x_i}$.

The quantum solutions of a linear system are controlled by representations of a certain group associated to the system:

\begin{defn}[Cleve, Liu, Slofstra~\cite{solgroup}]
Let $Mx = b$ be a linear system over $\Z_2$. Its \emph{solution group} is the finitely presented group $\Gamma$ with generators $g_1,\ldots,g_n$ and $J$ subject to the relations:
\begin{itemize}
\item $g_i^2 = 1$ for all $i$;
\item If $x_i$ and $x_j$ appear in the same equation, then $g_i g_j = g_j g_i$;
\item For each equation of the form $x_{k_1} + \ldots + x_{k_r} = b_r$, the generators satisfy
\[
	g_{k_1} \cdots g_{k_r} = J^{b_r}.
\]
\end{itemize}
\label{solgroupdef}
\end{defn}

These relations are precisely such that quantum solutions of $Mx=b$ on a Hilbert space $\H$ are in bijective correspondence with those unitary representations $\pi : \Gamma \to \U(\H)$ for which $\pi(J) = -\mathbbm{1}$.

Slofstra~\cite[Theorem~3.1]{tp} has shown that every finitely presented group embeds into a solution group in a particular way. Concerning undecidability, the following essential result was derived in the proof of~\cite[Corollary~3.3]{tp}.

\begin{thm}[Slofstra]
\label{solgroupundec}
Given a linear system $Mx = b$, it is undecidable to determine whether $J = 1$ in the associated solution group. Equivalently, it is undecidable to determine whether the linear system has a quantum solution.
\end{thm}

We now move on to considering the ramifications of this result, first for the hypergraph approach to contextuality~\cite{AFLS} and then for quantum logic, including the proof of Theorem~\ref{qlundec}.

\section{Consequences for the hypergraph approach to contextuality}

For us, a \emph{hypergraph} is a pair $H=(V,E)$ consisting of a finite set of vertices $V$ and a subset $E\subseteq 2^V$ with $\cup E = V$. Their relevance lies in the observation that hypergraphs provide a convenient and powerful language to analyze quantum contextuality~\cite{AFLS}:

\begin{defn}
A \emph{quantum representation} of a hypergraph $H=(V,E)$ consists of:
\begin{itemize}
\item A Hilbert space $\H$ with $\dim(\H)>0$,
\item a family of projections $(P_v)_{v\in V}$ in $\H$ assigned to the vertices of $H$ such that for each edge $e\in E$, the associated projections form a partition of unity,
\beq
\label{pueq}
\sum_{v\in e} P_v = \mathbbm{1}.
\eeq
\end{itemize}
\end{defn}

In the case where all projections have rank $1$, this is related to the notion of Kochen-Specker configuration: a finite collection of vectors in a Hilbert space such that certain particular subsets of these vectors form orthonormal bases. The concept of \emph{quantum model}~\cite[Definition~5.1.1]{AFLS} on a hypergraph---considered as a contextuality scenario---is implicitly based on our notion of quantum representation. In general, the idea is that the hyperedges $e\in E$ label measurements with outcomes $v\in V$, and some outcomes may be shared between several measurements, corresponding to vertices being incident to several hyperedges. 

\begin{ex}
A quantum representation of the hypergraph
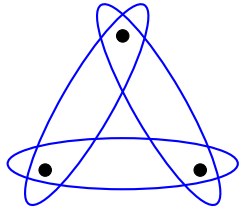
\begin{center}
\begin{tikzpicture}[scale=.85]
\node[draw,shape=circle,fill,scale=.5] (a) at (90:1.4) {} ;
\node[draw,shape=circle,fill,scale=.5] (c) at (210:1.4) {} ;
\node[draw,shape=circle,fill,scale=.5] (e) at (330:1.4) {} ;
\node[below of=a,node distance=3mm] {};
\node[above right of=c,node distance=3mm] {};
\node[above left of=e,node distance=3mm] {};
\draw[thick,blue,rotate=270] (0:.6) ellipse (.4cm and 1.8cm) ;
\draw[thick,blue,rotate=150] (0:.65) ellipse (.4cm and 1.8cm) ;
\draw[thick,blue,rotate=30] (0:.65) ellipse (.4cm and 1.8cm) ;
\end{tikzpicture}
\end{center}
would yield a nontrivial solution to the antecedents of Example~\ref{triangle}. Thus this hypergraph does not have any quantum representation.
\end{ex}

The quantum representations of a hypergraph are equivalently given by the representations of the \emph{free hypergraph C*-algebra}, which is the finitely presented C*-algebra
\[
	C^*(H) \defin \left\langle P_v :\: v\in V\:\bigg|\: P_v = P_v^* = P_v^2,\;\:\sum_{v\in e} P_v = \mathbbm{1} \:\:\forall e\in E\right\rangle,
\]
as already introduced in~\cite[Section~8.3]{AFLS}.
We quickly record a standard observation for future reference:

\begin{fact}
\label{pumo}
Projections that form a partition of unity~\eqref{pueq} are mutually orthogonal, $P_v P_w = \delta_{v,w} P_v$.
\end{fact}

Our central new observation is this:

\begin{lem}
There is an algorithm to compute, for every linear system $Mx = b$, a hypergraph $H$ such that quantum solutions of the linear system are in bijective correspondence with quantum representations of the hypergraph.
\label{algiso}
\end{lem}

This observation should not be surprising, since the original considerations around linear systems and solution groups~\cite{bcsg,solgroup} were inspired by the Mermin-Peres \emph{magic square}, one of the most startling examples of quantum contextuality~\cite{mermin}. The hypergraph construction in the following proof is along the lines of the \emph{measurement protocols} of~\cite[Appendix~D]{AFLS}, combined with forming an \emph{induced subscenario}~\cite[Definition~2.5.1]{AFLS}.

\begin{proof}
Let a linear system $Mx = b$ be given, with $M\in\Z_2^{m\times n}$ and $b\in\Z_2^m$. We write $[n]\defin \{1,\ldots,n\}$, and
\[
	N(r) = \left\{\: i \in [n]\;|\: M_{r,i} = 1 \:\right\}
\]
for the set of variables that are contained in the $r$-th equation, with $r\in[m]$. We think of $i\in[n]$ as indexing an observable $A_i$ with values in $\mathbf{2}\defin\{-1,+1\}$, and each $N(r)$ as indexing a set of measurements that commute and is therefore jointly implementable. Correspondingly, our hypergraph contains two kinds of outcome-representing vertices,
\[
	V \defin \{\: v_i^\alpha\: :\: i\in [n],\: \alpha\in\mathbf{2} \:\} \:\cup\: \{\: w_r^\beta \: :\: r\in [m],\: \beta \in \mathbf{2}^{N(r)}_{\pm} \:\},
\]
where the set $\mathbf{2}^{N(r)}_{\pm}$ consists of all those functions $\beta : N(r)\to\mathbf{2}$ which have the correct parity in the sense that $\prod_{i\in N(r)} \beta(i) = (-1)^{b_r}$.

The hyperedges will also be of three kinds: first, $\{v_i^{-1}, v_i^{+1}\}$ for every $i\in[n]$, which is intended to correspond to a measurement of the $\mathbf{2}$-valued observable $A_i$; second, $\{ w_r^\beta\: :\: \beta\in \mathbf{2}^{N(r)}_\pm\}$ for every $r\in[m]$, which corresponds to the possible outcomes of a joint measurement of the $A_i$ with $i\in N(r)$; and third, the sets of the form
\beq
\label{sedge}
	\{\, v_i^\alpha \,\} \:\cup\{\, w_r^\beta \: :\: \beta(i) = -\alpha\, \}
\eeq
for every fixed $r\in[m]$, $i\in N(r)$ and $\alpha\in\mathbf{2}$. Intuitively, this is the set of outcomes of the measurement protocol which first measures $A_i$, and if the outcome is $-\alpha$, then also conducts a joint measurement of all other observables $A_j$ with $j\in N(r)$, resulting in a joint outcome $\beta\in \mathbf{2}^{N(r)}_{\pm}$ with $\beta(i) = -\alpha$. When considered as one overall measurement, this protocol can either terminate after the first because $A_i$ gives $\alpha$, corresponding to the vertex $v_i^\alpha$; or it can continue to the second step, which produces some $\beta$ with $\beta(i) = - \alpha$, corresponding to one of the vertices $w_r^\beta$. This ends the definition of the relevant hypergraph $H$.

We now show how quantum representations of this hypergraph correspond to quantum solutions of the linear system. First, given a quantum representation of the hypergraph, we obtain a quantum solution of the linear system by taking
\beq
\label{PtoA}
 	A_i \defin P_{v_i^{+1}} - P_{v_i^{-1}} = 2 P_{v_i^{+1}} - \mathbbm{1}
\eeq
for all $i\in[n]$. We need to show that this indeed results in a quantum solution by checking that the $P_{v_i^{+1}} - P_{v_i^{-1}}$ are unitary and satisfy the required relations. Unitarity is clear since $2 P_{v_i^{+1}} - 1$ is a symmetry, which thereby also shows that the relation $A_i^2 = 1$ is respected. For the commutativity $A_i A_j = A_j A_i$ with $i,j\in N(r)$, we use Fact~\ref{pumo} together with the partition of unity relation associated to the third kind of edge~\eqref{sedge}: this relation implies that 
\beq
\label{extra}
	P_{v_i^{+1}} - P_{v_i^{-1}} = \sum_{\beta \: :\: \beta(i) = +1} P_{w_r^\beta} - \sum_{\beta \: :\: \beta(i) = -1} P_{w_r^\beta},
\eeq
and similarly for $j$. Hence applying Fact~\ref{pumo} to the partition of unity relation associated to the hyperedge $\{w_r^\beta\}$ shows that both $A_i$ and $A_j$ are linear combinations of the same pairwise orthogonal projections, which implies commutativity. The expression~\eqref{extra} is also useful for checking that the relation $\prod_{i\in N(r)} \big(P_{v_i^{+1}} - P_{v_i^{-1}}\big) = (-1)^{b_r}\mathbbm{1}$ holds as well, which then follows from $\sum_\beta P_{w_r^\beta} = \mathbbm{1}$ upon using that every $\beta$ has parity $(-1)^{b_r}$.

In the other direction, we put
\[
	P_{v_i^\alpha} \defin \frac{\mathbbm{1} + \alpha A_i}{2},\qquad P_{w_r^\beta} \defin \prod_{i\in N(r)} \frac{\mathbbm{1} + \beta(i) A_i}{2},
\]
where we likewise need to check that the relations are preserved, which first requires showing that both right-hand sides are projections. This is clear in the first case and holds by the commutativity assumption $A_i A_j = A_j A_i$ for $i,j\in N(r)$ in the second case. We verify the required partition of unity relations. First, $P_{v_i^{-1}} + P_{v_i^{+1}} = 1$ holds trivially. Second, if we apply the definition of $P_{w_r^\beta}$ also for $\beta$ of the wrong parity, then
\[
	\sum_{\beta\in\mathbf{2}^{N(r)}} P_{w_r^\beta} = \prod_{i\in N(r)} \,\sum_{\beta\in\mathbf{2}} \frac{\mathbbm{1} + \beta A_i}{2} = \mathbbm{1}.
\]
Since $P_{w_r^\beta} = 0$ whenever $\beta$ has the wrong parity, we can ignore these terms and arrive at the desired equation. Third, the relation associated to~\eqref{sedge} takes a bit more work: the expression $P_{v_i^\alpha} + \sum_{\beta\: :\: \beta(i) = -\alpha} P_{w_r^\beta}$ evaluates to
\[
	\frac{\mathbbm{1}+\alpha A_i}{2} + \sum_{\beta\: :\: \beta(i) = -\alpha} \:\prod_{j\in N(r)} \frac{\mathbbm{1} + \beta(j) A_j}{2}  = \frac{\mathbbm{1}+\alpha A_i}{2} + \frac{\mathbbm{1}-\alpha A_i}{2} \sum_{\beta\: :\: \beta(i) = -\alpha} \:\prod_{j\in N(r),\: j\neq i} \frac{\mathbbm{1} +\beta(j) A_j}{2}. 
\]
Upon expanding the product, the sum over $\beta$ makes all terms cancel except for the constant one and the $\prod_{j\neq i} A_j$ one, which survives as well due to the parity constraint on $\beta$. Therefore we arrive at the expression
\[
	\frac{\mathbbm{1}+\alpha A_i}{2} + \frac{\mathbbm{1}-\alpha A_i}{2} \cdot \frac{\mathbbm{1} - \alpha (-1)^{b_r} \prod_{j\neq i} A_j}{2} = \frac{\mathbbm{1}+\alpha A_i}{2} + \frac{\mathbbm{1}-\alpha A_i}{2} \cdot \frac{\mathbbm{1} - \alpha A_i}{2} = \mathbbm{1},
\]
and we have verified that all three types of relations hold.

We finally show that the previous two constructions are inverses of each other. Starting with a quantum solution $(A_i)_{i\in[n]}$, it is immediate to show that computing the resulting projections and using~\eqref{PtoA} results in the original $A_i$'s. A similar statement holds for the $P_{v_i^\alpha}$ in the other direction, while a short computation is required to show the same for the $P_{w_r^\beta}$,
\[
	\prod_{i\in N(r)} \frac{1 + \beta(i)\cdot \left(P_{v_i^{+1}} - P_{v_i^{-1}}\right)}{2} = \prod_{i\in N(r)} \left(1 - P_{v_i^{-\beta(i)}}\right) = \prod_{i\in N(r)} \: \sum_{\beta'\: :\: \beta'(i) = \beta(i)} P_{w_r^{\beta'}} = P_{w_r^\beta},
\]
where the last step again uses Fact~\ref{pumo}.
\end{proof}

In~\cite[Section~8]{AFLS}, we had considered the decision problem \texttt{ALLOWS\_QUANTUM}, which asks: given a hypergraph $H=(V,E)$, does it have a quantum representation?\footnote{The formulation of~\cite{AFLS} asks for the existence of a quantum model on $H$, but this is clearly equivalent: the existence of a quantum model requires the existence of a quantum representation to begin with; conversely, one can use a quantum representation and an arbitrary state in its underlying Hilbert space to obtain a quantum model.} Our \emph{inverse sandwich conjecture} hypothesized that this problem is undecidable. Thanks to Slofstra's Theorem~\ref{solgroupundec}, we are now in a position to prove this:

\begin{cor}[Inverse sandwich conjecture]
\label{isc}
There is no algorithm to determine whether a given hypergraph has a quantum representation.
\end{cor}

\begin{proof}
If there was such an algorithm, then Lemma~\ref{algiso} would provide an algorithm to determine whether a given linear system has a quantum solution. This is in contradiction with Theorem~\ref{solgroupundec}.
\end{proof}

This also implies that there are hypergraphs that have quantum representations, but only in infinite Hilbert space dimension~\cite[Section~8]{AFLS}. Translating Slofstra's explicit example~\cite[Corollary~3.2]{tp} into a hypergraph using the prescription of Lemma~\ref{PtoA} will provide an explicit (but large) example.

\begin{rem}
\label{puralg}
In terms of fancier language, one can phrase Lemma~\ref{algiso} as saying that the maximal group C*-algebra~\cite{groupCstar} of the solution group associated to the linear system is, after taking the quotient by the relation $J=-1$, computably isomorphic to a free hypergraph C*-algebra.
At the purely algebraic level of \emph{$*$-algebras}, the analogous statement is still true with the same proof, but one needs to throw in the orthogonality relations of Fact~\ref{pumo} separately when defining the finitely presented $*$-algebra associated to a hypergraph. We refer to~\cite{freehyper} for more details on the structure of free hypergraph C*-algebras and the corresponding $*$-algebras.
\end{rem}

\begin{cor}
\label{notrfd}
There are infinitely many hypergraphs $H$ for which $C^*(H)$ is not residually finite-dimensional.
\end{cor}

\begin{proof}
If every $C^*(H)$ was residually finite-dimensional, then we could use the algorithm of~\cite{compnorm} to determine whether $\|1\|=0$ or $\|1\|=1$ in $C^*(H)$, which are the only two possibilities depending on whether $C^*(H) = 0$ (no quantum representation) or $C^*(H)\neq 0$ (there is a quantum representation). This contradicts Corollary~\ref{isc}. If there were only finitely many exceptions to this residual finite-dimensionality, then there would also have to exist an algorithm which simply treats these exceptional cases separately.
\end{proof}

As far as we know, this is the first time that the strategy of~\cite{compnorm} has been successfully employed to show that some finitely presented C*-algebras are not residually finite-dimensional. We do not know \emph{which ones} of these C*-algebras fail to be residually finite-dimensional.

\section{Consequences for quantum logic}
\label{mainthm}

The projection operators on a Hilbert space are in bijective correspondence with the closed subspaces. This lets us translate the observations of the previous section into statements about quantum logic.

Although we try to avoid too much jargon, it will be helpful to utilize the basic terminology of model theory~\cite{mtbook}. We work in the signature $(\lor,\perp,0,1)$, where $\perp$ is a binary relation; any orthomodular lattice can also be considered a structure of this signature. We follow~\cite{mtbook} in using the shorthand notation $\bar{P}\defin P_1,\ldots, P_n$ for a list of variables, using notation which suggests that we are still thinking in terms of projections.

It is a standard fact that forming a partition of unity~\eqref{pueq} at the level of projections is equivalent to the associated subspaces being pairwise orthogonal and spanning the entire space, which translates into the formula
\beq
\label{PUeq}
\PU{\bar{P}} = \PU{P_1,\dots,P_n} \defin \mybigand_{i\neq j} \left( \, P_i \perp P_j \, \right) \, \myand \, \left( P_1 \lor \ldots \lor P_n = 1 \right) ,
\eeq
where the notation `$OC$' reminds us of \emph{orthogonality} and \emph{completeness}. All formulas that we use are built out of formulas of this form. If we want to say that already a certain subset $\{P_i \: :\: i\in e\}$ of these projections indexed by $e\subseteq[n]$ forms a partition of unity, then we simply write $\PU{\bar{P}_e}$.

\begin{thm}
The theory of complex Hilbert lattices $\Cl(\H)$ in the signature $(\lor,\perp,0,1)$ is undecidable: there is no algorithm to decide whether for a given hypergraph $H=(V,E)$, the sentence
\beq
\label{horn}
\forall (P_v)_{v\in V} \, \left[ \Big(\mybigand_{e\in E} \PU{\bar{P}_e} \Big) \myimplies \big(0 = 1) \right]
\eeq
holds in every $\Cl(\H)$ or not.
\label{sandwich}
\end{thm}

This clearly implies our \Cref{qlundec}.

\begin{proof}
This is now merely a restatement of Corollary~\ref{isc}: the tuples of projections in a Hilbert space $\H$ with $\dim(\H)>0$ such that $\mybigand_{e\in E} \PU{\bar{P}_e}$ holds are precisely the quantum representations of $H$, while $0=1$ is equivalent to $\dim(\H) = 0$.
\end{proof}

Finally, we show that the universal theory of complex Hilbert lattices does not really depend on the Hilbert space. To this end, note that also the lattice of projections of any von Neumann algebra is a structure in the signature\footnote{See~\cite{redei} for an introduction to von Neumann algebras from the perspective of quantum logic, including a treatment of their lattices of projections (Section 6.2).} $(\lor,\perp,0,1)$. For example if this von Neumann algebra is $\B(\H)$, the algebra of bounded operators on a Hilbert space $\H$, then this structure is exactly the Hilbert lattice $\Cl(\H)$. If the von Neumann algebra is just $\C^n$, then this structure is the Boolean algebra of subsets of $[n]$, where two subsets are considered orthogonal if and only if they are disjoint. In order to understand our conclusion, the reader only need to know that every $\B(\H)$ is a von Neumann algebra, and that every algebra homomorphism $\B(\H) \to \B(\H \otimes \H')$ given by $a \mapsto a \otimes 1$ is a normal\footnote{Where as usual in von Neumann algebra theory, \emph{normal} means continuous with respect to the ultraweak topologies.} injective $*$-homomorphism.

\begin{lem}
Let $\phi(\bar{P})$ be a quantifier-free formula in the signature $(\lor,\perp,0,1)$. 
\begin{enumerate}
\item\label{vNaembed} If $f:\mathcal{N}\to\mathcal{M}$ is a normal injective $*$-homomorphism between von Neumann algebras, then $\phi(f(P_1),\ldots,f(P_n))$ if and only if $\phi(P_1,\ldots,P_n)$.
\item\label{b} Let $\Phi$ be the universal sentence $\forall \bar{P}\; \phi(\bar{P})$.
Then $\mathcal{H}\vDash\Phi$ for all Hilbert spaces $\mathcal{H}$ if and only if $\mathcal{H}\vDash\Phi$ for separable infinite-dimensional $\mathcal{H}$.
\end{enumerate}
\label{allH}
\end{lem}

\begin{proof}
\begin{enumerate}
	\item Normality of $f$ enters because it guarantees that $f(P_1\lor P_2) = f(P_1)\lor f(P_2)$; one way to see this is to choose a faithful normal representation $\mathcal{M} \subseteq \B(\H)$ for some Hilbert space $\H$, which upon composing with $f$ also gives a faithful normal representation of $\mathcal{N}$ in $\mathcal{B}(\H)$. Then since a faithful normal representation induces a lattice homomorphism between the projection lattices~\cite[Proposition~6.3]{redei}, we can conclude the same for $f$ as well. We likewise have $P_1 \perp P_2 \: \Leftrightarrow \: f(P_1) \perp f(P_2)$, which is clear also merely by multiplicativity and injectivity of $f$. We trivially have $f(0) = 0$ and $f(1) = 1$. Therefore $f$ is an embedding of structures in the signature $(\lor,\perp,0,1)$. The claim is then standard, and can also be proven easily by induction on the complexity of $\phi$.
\item 
The nontrivial direction is this: if $\mathcal{H}\vDash\Phi$ for separable infinite-dimensional $\mathcal{H}$, then also $\mathcal{H}'\vDash\Phi$ for any other Hilbert space $\mathcal{H}'$.
\begin{enumerate}
\item [\underline{Case 1:}] $\H'$ is finite-dimensional. In this case, $\mathcal{H}'\otimes\mathcal{H}$ is isomorphic to $\mathcal{H}$, and therefore $\mathcal{H}'\otimes\mathcal{H}\vDash \Phi$. By assumption, we therefore know that $\phi(P_1\otimes\mathbbm{1},\ldots,P_n\otimes\mathbbm{1})$ for any tuple of projections $\bar{P}$ on $\mathcal{H}'$, and hence also $\phi(P_1,\ldots,P_n)$ by~\ref{vNaembed}.
\item [\underline{Case 2:}] $\mathcal{H}'$ is infinite-dimensional. In this case, any tuple of projections $\bar{P}$ generate a von Neumann subalgebra in $\mathcal{B}(\mathcal{H}')$ that we denote $\mathcal{N}$. By~\cite[Theorem~2.1]{sherman}, there is a direct sum decomposition $\mathcal{N} \cong \bigoplus_{i \in I} \mathcal{N}_i$ such that each $\mathcal{N}_i$ is separable, and can therefore be represented faithfully on $\mathcal{H}$.

	Now let $\bar{P}$ be a tuple of projections on $\H'$. Our goal is to show that $\phi(\bar{P})$ holds in $\mathcal{B}(\H')$, or equivalently in $\mathcal{N}$. With $\bar{P}^{(i)}$ for $i \in I$ denoting the tuple of projections projected to the direct summand $\mathcal{N}_i$, the assumption implies that $\phi(\bar{P}^{(i)})$ holds in $\mathcal{N}_i$ for every $i$. But since the truth of a quantifier-free formula can be determined componentwise in a direct sum---as is easily seen by induction on the complexity---it follows that $\phi(\bar{P})$ indeed holds in $\mathcal{N}$. \qedhere
\end{enumerate}
\end{enumerate}
\end{proof}

\begin{rem}
	Thus the universal (first-order) theory of the projection lattice of any separable infinite-dimensional Hilbert space like $\ell^2(\N)$ coincides with the (first-order) theory of all complex Hilbert lattices. \Cref{sandwich} thus implies that the theory of a separable infinite-dimensional Hilbert space is undecidable as well. Since it is also trivially complete, it follows that it is not even recursively axiomatizable: if one could enumerate the axioms recursively, then one would have an algorithm to deciding for any sentence $\Phi$ whether $\Phi$ or $\lnot\Phi$ holds by enumerating the axioms together with all of their consequences until $\Phi$ or $\lnot\Phi$ has been derived.

	One can also prove that the theory of all complex Hilbert lattices cannot be axiomatized recursively, as hypothesized by Svozil in~\cite[p.\,69]{rand}. With ``the theory'' referring to the first-order theory of complex Hilbert lattices, this is based on the following observations:
	\begin{itemize}
			\setlength\itemsep{0.6\baselineskip}
		\item There is a universal sentence $\Psi$ such that $\H \models \Psi$ for every finite-dimensional Hilbert space $\H$, but $\H \models \lnot\Psi$ for separable infinite-dimensional $\H$.

			Indeed the projection lattice of finite-dimensional $\H$ satisfies the modular law, while already Birkhoff and von Neumann showed that this fails in $\ell^2(\N)$~\cite[p.~832]{logicqm}. In follow-up work to this paper, it was also shown that there are $\Psi$ of the form~\eqref{horn} with this property~\cite[Theorem~3.5]{freehyper}. 
		\item Suppose that we fix such $\Psi$. Then a universal sentence $\Phi$ satisfies $\ell^2(\N) \models \Phi$ if and only if $\Psi \lor \Phi$ belongs to the theory.

			Indeed in one direction, $\ell^2(\N) \models \Phi$ implies $\H \models \Phi$ for all infinite-dimensional $\H$ by the arguments used for \Cref{allH}, while trivially $\H \models \Psi$ for finite-dimensional $\H$. Conversely, if $\Psi \lor \Phi$ belongs to the theory, then we must have $\ell^2(\N) \models \Phi$ since $\ell^2(\N) \not\models \Psi$.
		\item Similarly, a universal sentence $\Phi$ satisfies $\ell^2(\N) \models \lnot \Phi$ if and only if $\Psi \lor \lnot \Phi$ belongs to the theory.

			Indeed $\ell^2(\N) \models \lnot \Phi$ implies $\H \models \lnot \Phi$ for any infinite-dimensional $\H$, by proving the contrapositive as in \Cref{allH}, while trivially $\H \models \Psi$ for finite-dimensional $\H$. Conversely, if $\Psi \lor \lnot \Phi$ belongs to the theory, then again $\ell^2(\N) \models \lnot \Phi$ since $\ell^2(\N) \not\models \Psi$.

	\end{itemize}
	It follows that for every universal sentence $\Phi$, either $\Psi \lor \Phi$ or $\Psi \lor \lnot\Phi$ belongs to the theory, but not both. Now if one could enumerate the axioms recursively, then one would have an algorithm for deciding which of these two cases occurs for any given $\Phi$, again by generating all axioms together with all their consequences until either $\Psi \lor \Phi$ or $\Psi \lor \lnot\Phi$ has been derived. Therefore no recursive axiomatization of the theory of complex Hilbert lattices exists.

Svozil~\cite[p.\,69]{rand} also asks if it is possible to develop an axiomatization of Hilbert lattices in ``purely algebraic'' terms. If one interprets this as asking whether the Hilbert lattices are the class of models of some theory in first-order logic, then the answer is well-known to be negative: the L\"owenheim-Skolem theorem asserts that it is impossible to axiomatize \emph{any} uncountable structure in first-order logic. This holds irrespectively of whether one attempts to axiomatize Hilbert lattices in all Hilbert space dimensions or only in one particular dimension $\geq 2$. The most that one can hope for is categoricity in the relevant cardinality, meaning that every model of the same cardinality as the intended model is isomorphic to the intended models. But there will always be models in other cardinalities as well.
\end{rem}

\subsection*{Acknowledgements}

We thank John Harding, Christian Herrmann, Ravi Kunjwal, Anthony Leverrier, Mladen Pavi{\v{c}}i{\'c}, Stefan Schmidt, William Slofstra, Rob Spekkens, Karl Svozil, Andreas Thom and Moritz Weber for discussions. Special thanks go to Christian Herrmann for copious help with the literature and useful feedback, to William Slofstra for pointing out that Theorem~\ref{sandwich} actually follows from the arguments used in an earlier version of this paper to prove a weaker result, and to Andre Kornell for pointing out two gaps in our arguments in a previous version (which unfortunately is the published version).

Most of this paper was written while the author was with the Max Planck Institute for Mathematics in the Sciences in Leipzig, Germany.

\bibliographystyle{unsrt}
\bibliography{quantum_logic}

\bigskip

\end{document}